\documentclass{article}

\PassOptionsToPackage{numbers}{natbib}
\usepackage[final,sglblindworkshop]{neurips_2025}
\workshoptitle{First Workshop on LLM Persona Modeling}

\usepackage[utf8]{inputenc} % allow utf-8 input
\usepackage[T1]{fontenc}    % use 8-bit T1 fonts
\usepackage{hyperref}       % hyperlinks
\usepackage{url}            % simple URL typesetting
\usepackage{booktabs}       % professional-quality tables
\usepackage{amsfonts}       % blackboard math symbols
\usepackage{nicefrac}       % compact symbols for 1/2, etc.
\usepackage{microtype}      % microtypography
\usepackage{xcolor}         % colors
\usepackage{graphicx}       % for figures
\usepackage{amsmath}        % for math environments
\usepackage{amssymb}        % for math symbols
\usepackage{amsthm}         % for theorems
\usepackage[most]{tcolorbox} % for colored boxes

\newtheorem{theorem}{Theorem}
\newtheorem{corollary}{Corollary}

\newtheorem{hypothesis}{Hypothesis}
\newtheorem*{theorem*}{Theorem}

\title{Love First, Know Later: Persona-Based Romantic Compatibility Through LLM Text World Engines}

\author{%
  \textbf{HaoYang Shang}$^{1}$\thanks{Equal contribution.} 
  \And
  \textbf{Zhengyang Yan}$^{1}$\footnotemark[1] 
  \And
  \textbf{Xuan Liu}$^{1}$\thanks{Corresponding author: \texttt{info.breathingcore@gmail.com}}
  \AND 
  \normalfont $^{1}$BreathingCORE \\ % <--- Added \normalfont here!
}

\begin{document}

\maketitle

\begin{abstract}
We propose \textbf{\textit{Love First, Know Later}}: a paradigm shift in computational matching that simulates interactions first, then assesses compatibility. Instead of comparing static profiles, our framework leverages LLMs as \textbf{\textit{text world engines}} that operate in dual capacity—as persona-driven agents following behavioral policies and as the environment modeling interaction dynamics. We formalize compatibility assessment as a reward-modeling problem: given observed matching outcomes, we learn to extract signals from simulations that predict human preferences. Our key insight is that relationships hinge on responses to critical moments—we translate this observation from relationship psychology into mathematical hypotheses, enabling effective simulation. Theoretically, we prove that as LLM policies better approximate human behavior, the induced matching converges to optimal stable matching. Empirically, we validate on speed dating data for initial chemistry and divorce prediction for long-term stability. This paradigm enables interactive, personalized matching systems where users iteratively refine their agents, unlocking future possibilities for transparent and interactive compatibility assessment.
\end{abstract}
\section{Introduction}

Modern dating platforms primarily rely on profile similarity metrics, yet decades of relationship research show that compatibility emerges from interaction dynamics, not static attributes \cite{finkel2012online}. We propose \textbf{Love First, Know Later}: instead of comparing static profiles to predict compatibility, we simulate the relationship itself first.

Recent advances show that LLMs can express consistent personality traits \cite{park2023generative,li2024big5chat,wang2025incharacter}, follow human psychological patterns \cite{kosinski2023theory}, and engage in social interactions \cite{yang2024oasis,zhou2024sotopia,liu2024roleagent}. These models encode rich social priors from their training data, enabling them to simulate human behaviors and emotional responses. We leverage these capabilities to use LLMs as proxies for simulating romantic interactions. Building on the finding that LLMs possess extensive world knowledge and can maintain coherent personas across extended dialogues, we propose the concept of an \textbf{LLM text world engine}—a system where an LLM operates in dual capacity: (1) as \emph{agents} following behavioral policies $\pi_i, \pi_j$ derived from user personas and (2) as the \emph{environment} simulating conversation dynamics and state transitions (Figure \ref{fig:system_overview}). This dual-capacity design benefits from recent advances in LLM-based social simulation \cite{yang2024oasis,chang2025llms,liu2024training}, role-playing agent consistency \cite{liu2024roleagent,xu2025persona}, multi-agent coordination \cite{shang2025united}, social alignment and cognition \cite{liu2024training,liu2025synthetic,liuexploring}, and controllable agent behaviors \cite{liu2025programmable}, creating coherent simulation spaces for social behaviors.

\begin{figure}[h]
\centering
\includegraphics[width=0.8\textwidth]{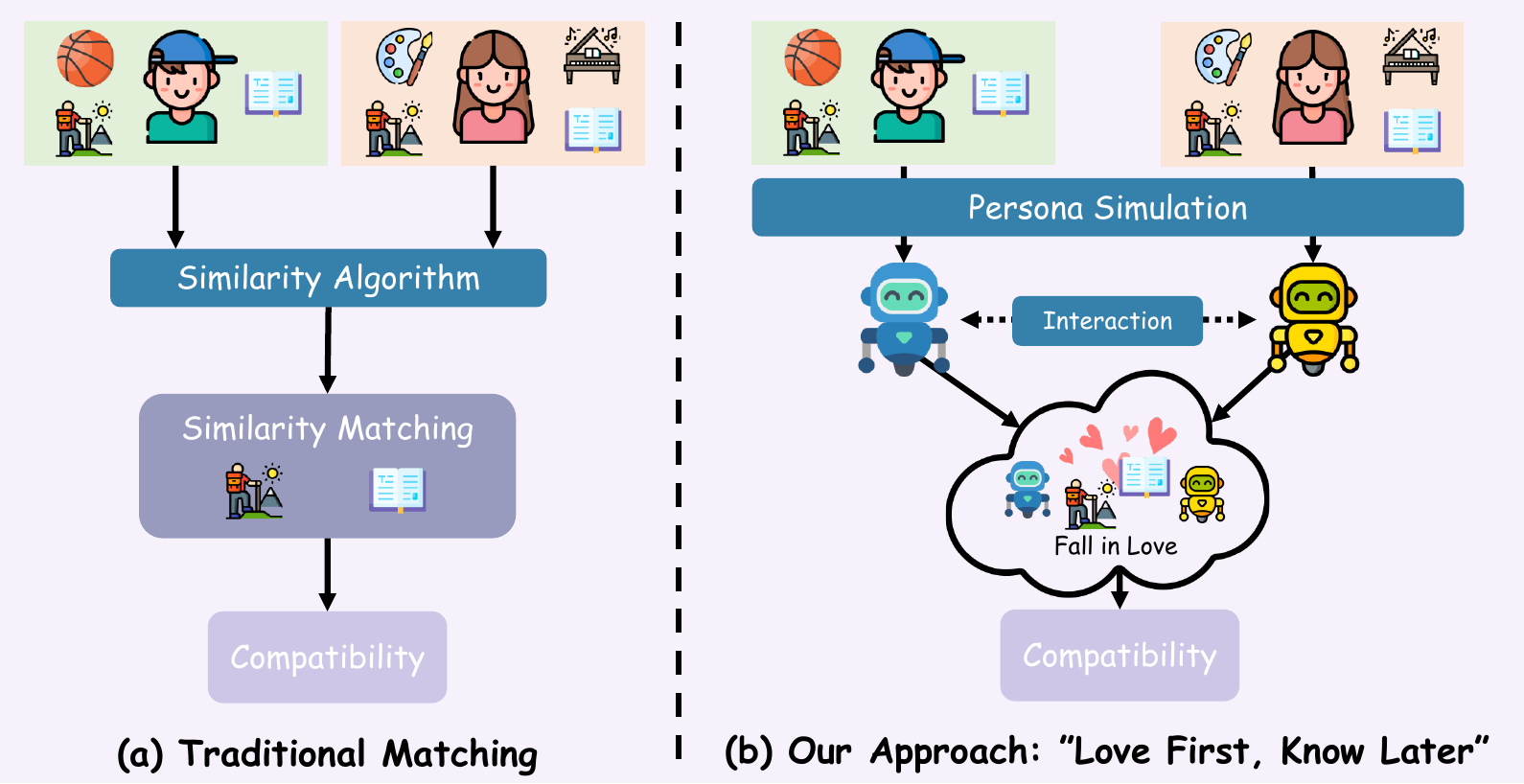}
% \fbox{\rule[-.5cm]{0cm}{2.5cm} \rule[-.5cm]{12cm}{0cm}}
\caption{\textbf{System Overview}: Traditional matching (left) compares static profiles to compute similarity. Our approach (right) uses LLM text world engines to simulate relationship dynamics, with the LLM operating in dual mode—as persona-driven agents and as the environment modeling interaction flow.}
\label{fig:system_overview}
\end{figure}

We formalize compatibility assessment as a \textbf{reward modeling problem} with an inverse RL flavor \cite{ng2000algorithms,christiano2017deep}, and connect it to preference learning problem \cite{rafailov2023direct,ramesh2024group}: given observed matching outcomes and fixed LLM policies approximating human behavior, we learn to extract signals from simulated interactions that predict real human preferences. This paradigm shift—from profile comparison to interaction simulation—enables us to capture the emergent properties of relationships that static features cannot represent. Theoretically, we prove that as LLM policies better approximate human behavior, compatibility predictions converge to optimal matching. Empirically, we initially validate our approach on speed dating and divorce prediction datasets. 

\paragraph{Contributions.} Our work makes three key contributions: \textbf{(1) Paradigm and Architecture:} We propose the \emph{Love First, Know Later} paradigm—simulating interactions to assess compatibility—and introduce \emph{LLM text world engines} that operate in dual capacity as both agents and environment. \textbf{(2) Mathematical Formalization with Theoretical Guarantees:} We formalize compatibility as a reward modeling problem and translate observations from relationship psychology and social science into mathematical hypotheses, bridging qualitative social insights with quantitative frameworks. Under these hypotheses, we prove convergence of our methods to optimal matching outcome. \textbf{(3) Empirical Validation and Analysis:} We validate feasibility through experiments on speed dating and divorce prediction, and we provide an in-depth analysis of the potential future of our paradigm.
% \footnote{Code and data: \url{https://anonymous.4open.science/r/test-love-first-D111}}

\section{Method: Reward Modeling for Compatibility}

\subsection{Problem Formulation}

We formalize romantic compatibility prediction as a \textbf{reward modeling problem} \cite{christiano2017deep}. Given a dataset $\mathcal{D} = \{(\pi_i, \pi_j, y_{ij})\}$ where $\pi_i, \pi_j$ are individuals' policies and $y_{ij} \in \{0,1\}$ indicates mutual matching, our goal is to learn a reward function $R: \Pi \times \Pi \rightarrow \mathbb{R}$ predicting human preferences (see Appendix \ref{app:mdp} for MDP formulation). This builds on inverse reinforcement learning \cite{ng2000algorithms} and reward modeling from human feedback \cite{christiano2017deep}. One issue is we don't have human policies $\pi_i^*$ and $\pi_j^*$. We approximate them using LLMs: $\hat{\pi}_i(a|s) = \text{LLM}(a|s, P_i)$ based on persona $P_i$, then extract features from simulation to learn human preferences.

\subsection{LLM Text World Engine for Interaction Simulation}

For each pair of individuals $(i,j)$, we deploy an LLM text world engine that operates in dual capacity. First, it simulates human behavioral policies: $\hat{\pi}_i(a|s) = \text{LLM}(a|s, P_i)$ and $\hat{\pi}_j(a|s) = \text{LLM}(a|s, P_j)$ based on personas $P_i$ and $P_j$. Second, it simulates the environment by generating interaction topics, modeling emotional state transitions, and determining how the dialogue context evolves based on participants' responses. With both agents and environment simulated, we obtain a complete interaction trajectory $\tau_{ij} = \{(s_0, a_i^0, a_j^0), ..., (s_T, a_i^T, a_j^T)\}$ where states capture the evolving context and actions represent agents response. The quality of our compatibility prediction critically depends on how well $\hat{\pi}_i$ approximates the true human policy $\pi_i^*$.

\subsection{Learning Human Preferences via Love Observer}

We propose the idea of \emph{love observer}—a specialized LLM that extracts ratings from simulated interactions: \textbf{Individual Participant Ratings.} After interaction $\tau_{ij}$, the observer assesses each participant's perspective independently, producing $r_1 = \text{Observer}(\tau_{ij}, p_i)$ and $r_2 = \text{Observer}(\tau_{ij}, p_j)$. \textbf{Observer Rating.} The observer also provides an external assessment $r_3 = \text{Observer}(\tau_{ij}, p_{\text{ext}})$ evaluating overall compatibility based on conversational flow, mutual engagement, and value alignment. We learn a compatibility score $R(i,j)$ by combining these ratings using a simple linear method trained on observed matching decisions. This approach is inspired by recent work on calibrating LLM judgments and learning from preference data \cite{shao2025earlier,slocum2025diverse}.

\section{Efficient Simulation via Critical Events}

Full relationship simulation is computationally intractable—years of interactions are impossible. We use two insights from relationship psychology that enable effective compatibility assessment.

\begin{hypothesis}[Sparse Rewards]
Relationship outcomes are determined by responses to a small number of critical moments. Formally, $R(s) \neq 0$ only for $s \in \mathcal{S}_{\text{critical}}$ where $|\mathcal{S}_{\text{critical}}| \ll |\mathcal{S}|$.
\end{hypothesis}

\begin{hypothesis}[Deterministic Decisions]
In critical moments, individuals exhibit consistent decision patterns. The policy entropy $H(\pi(\cdot|s)) < \delta$ for small $\delta$ when $s \in \mathcal{S}_{\text{critical}}$.
\end{hypothesis}

These hypotheses are grounded in relationship research showing that critical moments—conflict resolution \cite{gottman1998predicting}, first-date impressions \cite{finkel2012online}, value-alignment discussions \cite{ban2024neural}—strongly predict outcomes, while routine interactions contribute minimally (Hypothesis 1). Moreover, trait activation theory \cite{tett2003personality} and situational strength research \cite{meyer2010situational} show individuals exhibit consistent behaviors in trait-relevant situations (Hypothesis 2). This motivates our two simulation modes. \textbf{Speed Dating Mode} tests whether LLM policies can learn human preferences for initial chemistry through brief dialogues. \textbf{Critical Events Mode} validates our sparse rewards hypothesis by probing fundamental compatibility through several pivotal scenarios (career conflicts, family planning)—these rare but deterministic moments reveal true compatibility. We give a theorem ensuring that improving LLM agents directly translates to better matching outcomes—justifying investment in agent refinement.

\begin{theorem}[Convergence Guarantee]\label{thm:convergence_full}
Under above hypotheses, as the LLM policy approximation error $\epsilon \to 0$: (1) the prediction error $|\hat{R}(i,j) - R^*(i,j)|$ between predicted $\hat{R}$ and true reward $R^*$ vanishes, and (2) the induced matching converges to optimal stable matching (proof in Appendix \ref{app:theorem1}).
\end{theorem}

\section{Experimental Setup and Results}

We evaluate our approach on two relationship datasets capturing different compatibility aspects. The Columbia Speed Dating dataset tests initial chemistry prediction. The Divorce Prediction dataset validates our critical events hypothesis for long-term stability. To ensure fair comparison, all baselines underwent parameter sweeping to ensure optimal performance. Full details are in \ref{app:experimental_details}.

\begin{table}[h]
\centering
\begin{minipage}{0.48\textwidth}
\centering
\begin{tabular}{lcc}
\toprule
Method & F1 & AUC \\
\midrule
Logistic Regression & 0.66 & 0.61 \\
Similarity & 0.55 & 0.54 \\
\midrule
LLM Love Observer & 0.67 & 0.60 \\
LLM Mixed & 0.64 & 0.57 \\
\bottomrule
\end{tabular}
\end{minipage}
\hfill
\begin{minipage}{0.48\textwidth}
\centering
\begin{tabular}{lcc}
\toprule
Method & F1 & AUC \\
\midrule
Logistic Regression & 0.61 & 0.60 \\
Similarity & 0.65 & 0.50 \\
\midrule
LLM Love Observer & 0.67 & 0.56 \\
LLM Mixed & 0.67 & 0.57 \\
\bottomrule
\end{tabular}
\end{minipage}
\caption{Match prediction performance. \textbf{Left:} Stage 1 uses pre-dating information only. \textbf{Right:} Stage 2 adds during-date signals. LLM Mixed combines multiple LLM ratings using learned weights.}
\label{tab:rq1_results}
\end{table}

\paragraph{RQ1: Does Speed Dating Simulation Capture Initial Human Chemistry?}

Table~\ref{tab:rq1_results} compares our approach against logistic regression and similarity baselines. The LLM methods exceeds baselines. Performance is modest overall due to dataset sparsity, but simulated interactions extract meaningful signal.

\paragraph{RQ2: Do Critical Events Reveal Long-term Compatibility?}

We validate our critical events hypothesis on the Divorce Prediction dataset by generating personalized critical scenarios and simulating reactions. Unlike speed dating (limited information), this task has abundant correlated features (> $0.9$ correlation), making logistic regression very strong. Despite this challenging comparison, our observer method (10 ICL examples) performs comparably without personalized tuning, validating the framework's potential across relationship stages.

\begin{table}[h]
\centering
\begin{tabular}{lcc}
\toprule
Method & F1 & AUC \\
\midrule
Logistic Regression  & 0.95 & 1.00 \\
\midrule
LLM Love Observer & 0.90 & 0.92 \\
\bottomrule
\end{tabular}
\caption{Divorce Prediction Results: Critical Events Simulation vs Baseline}
\label{tab:divorce}
\end{table}

\section{Discussion}

We presented a paradigm shift in compatibility assessment: using LLM text world engines to simulate interactions, predicting outcomes. Our approach will capture the emergent dynamics of relationships that static features cannot represent. Key contributions include: (1) formalizing compatibility as a reward modeling problem, (2) introducing and justifying the hypothesis for effective simulation via critical events, and (3) Initial validation showing validity even without fine-tuning of base model. While our experiments demonstrate feasibility, the true significance lies in the future potential. 

\textbf{Revolutionary potential: From static to interactive matching.} Our framework reimagines matchmaking through four capabilities:

(i) \textbf{Personalized evolution and continuous improvement}: Each user develops their own agent that improves through feedback, creating personalized compatibility predictors. Users can refine agents through preference optimization techniques—ranging from training-based methods like DPO \cite{rafailov2023direct} and its variants (\cite{ramesh2024group}, \cite{shao2025earlier},\cite{slocum2025diverse}) to training-free approaches like In-Context DPO \cite{song-etal-2025-instantly} that operate purely through prompts (one individual agent per person). As users provide more feedback and as foundation models advance \cite{liuexploring,liu2025programmable}, predictions automatically improve, converging toward optimal matching.

(ii) \textbf{Bidirectional interaction and transparency}: Unlike black-box systems, users actively participate in the matching process through bidirectional interaction. They observe simulated interactions between their agent and potential matches, gaining transparent insights into compatibility assessments. Users provide feedback on agent behaviors and simulation quality, which refines both the agent policy and observer judgments. This continuous dialogue between user and system \cite{park2023generative,liu2024roleagent} ensures persona authenticity \cite{xu2025persona} while transforming matching from passive algorithm acceptance to active, collaborative exploration.

(iii) \textbf{Active preference exploration}: Agents probe hidden compatibility aspects through new scenarios, helping users discover what matters through simulated experiences rather than questionnaires. This draws on active learning \cite{ban2024neural} and exploration principles \cite{sutton2018reinforcement}. Unlike traditional methods where profiles are fixed and cannot reveal personal preferences on unasked dimensions, simulation enables dynamic exploration of unique relationship scenarios and concerns.

Despite this potential, limitations remain: (i) text-only simulations miss non-verbal cues but can be mitigated by using novel-like gesture/expression/tone description; (ii) our current dyadic modeling focuses on pairwise compatibility—extending to multi-person social dynamics (e.g., group dating, family integration, friend network compatibility) requires game-theoretic frameworks beyond simple pairwise matching. As LLMs advance in persona consistency \cite{xu2025persona,liu2024roleagent} and social reasoning \cite{liu2024training}, we expect this \textbf{Love First, Know Later} approach to unlock new possibilities wherever compatibility emerges through interaction.

\bibliographystyle{plainnat}
\bibliography{references}

\appendix

\section{Related Work}
\label{app:related_work}

\citet{wang2025incharacter} and \citet{li2024big5chat} demonstrate that LLMs can maintain consistent personality traits with high fidelity—validating our core assumption for romantic compatibility simulation. \citet{liu2025synthetic} show how persona dimensions affect social dynamics in multi-agent interactions, while \citet{zhou2024sotopia} provide evaluation frameworks for social intelligence. \citet{xiao2025humanizing} propose intentional anthropomorphic design in LLMs, supporting our vision of transparent user-agent interaction. \citet{occhipinti2024prodigy} introduce profile-based dialogue generation, complementing our persona-to-narrative transformation approach. We extend this foundation to romantic compatibility assessment with the critical events hypothesis and theoretical convergence guarantees.

\section{Proof of Theorem 1}
\label{app:theorem1}

\begin{theorem*}[Full Statement: Convergence to Optimal Matching]
Under the sparse rewards and deterministic decisions hypotheses, let $\epsilon = \max_{s,a} |\hat{\pi}_i(a|s) - \pi_i^*(a|s)|$ measure the policy approximation error and $\delta$ bound the policy entropy at critical states. Then:
\begin{enumerate}
\item The prediction error is bounded: $|\hat{R}(i,j) - R^*(i,j)| \leq L_\epsilon \epsilon + L_\delta \delta$ where $L_\epsilon, L_\delta$ depend on the number of critical states and observer Lipschitz constant.
\item As $(\epsilon, \delta) \to (0,0)$, the stable matching $\hat{\mathcal{M}}$ induced by predicted rewards converges to the optimal matching $\mathcal{M}^*$ under true rewards.
\end{enumerate}
\end{theorem*}

\begin{proof}
\textbf{Part 1: Error Bound}

Under our sparse rewards and deterministic decisions hypotheses, we focus on critical states $\mathcal{S}_{\text{critical}}$. Since the observer LLM is Lipschitz continuous with respect to the input conversation, the difference between the ratings extracted from the simulated interaction ($\mathbf{r}_{ij} = [r_1, r_2, r_3]^T$) and the true ones ($\mathbf{r}^*_{ij}$) is bounded:
$$||\mathbf{r}_{ij} - \mathbf{r}^*_{ij}|| \leq L_{\text{obs}} \cdot d(\tau_{ij}, \tau^*_{ij})$$
where $L_{\text{obs}}$ is the Lipschitz constant of the observer and $d(\cdot, \cdot)$ is a distance metric between interaction trajectories.

By the sparse rewards hypothesis, the critical contribution to the reward comes from $|\mathcal{S}_{\text{critical}}| \ll |\mathcal{S}|$ states. The distance between the simulated and true interactions at these critical states decomposes into two terms: a policy-approximation term and a sampling/concentration term due to single-rollout estimation under low entropy:
$$d(\tau_{ij}, \tau^*_{ij}) \leq C_\epsilon \cdot |\mathcal{S}_{\text{critical}}| \cdot \epsilon \; + \; C_\delta \cdot \delta,$$
for constants $C_\epsilon, C_\delta > 0$. Combining these, the error in the final predicted reward is
$$|\hat{R}(i,j) - R^*(i,j)| = |\mathbf{w}^T(\mathbf{r}_{ij} - \mathbf{r}^*_{ij})| \leq ||\mathbf{w}|| \cdot ||\mathbf{r}_{ij} - \mathbf{r}^*_{ij}|| \leq L_\epsilon\, \epsilon + L_\delta\, \delta,$$
where $L_\epsilon = ||\mathbf{w}|| \cdot L_{\text{obs}} \cdot C_\epsilon \cdot |\mathcal{S}_{\text{critical}}|$ and $L_\delta = ||\mathbf{w}|| \cdot L_{\text{obs}} \cdot C_\delta$. Since $|\mathcal{S}_{\text{critical}}| \ll |\mathcal{S}|$ and $\delta$ is small by hypothesis, this yields a tight bound compatible with single-trajectory estimation.

\textbf{Part 2: Convergence to Optimal Matching}

We consider the Gale-Shapley algorithm for finding a stable matching: A matching is stable if there is no unmatched pair $(i,j)$ who both prefer each other over their currently assigned partners. Let $\mathcal{M}^*$ be the stable matching under the true rewards $R^*$ and $\hat{\mathcal{M}}$ be the stable matching under our predicted rewards $\hat{R}$.

For any two potential partners $k$ and $j$ for individual $i$, if $R^*(i,j) > R^*(i,k)$, then for a sufficiently small $\epsilon$, our reward prediction bound ensures that the preference ordering is preserved:
$$\hat{R}(i,j) > R^*(i,j) - L\epsilon > R^*(i,k) - L\epsilon > \hat{R}(i,k) - 2L\epsilon$$
As $\epsilon \to 0$, the term $2L\epsilon$ vanishes, meaning $\hat{R}(i,j) > \hat{R}(i,k)$. Therefore, the preference ordering of our predicted rewards converges to the true preference ordering. By the uniqueness of the stable matching solution when preferences are distinct (which holds with probability 1 for continuous reward functions), it follows that for a sufficiently small $\epsilon$, $\hat{\mathcal{M}} = \mathcal{M}^*$.
\end{proof}

\begin{corollary}[Online Learning Reduces Approximation Error]
\label{cor:online_learning}
The policy approximation error $\epsilon$ can be reduced through online learning via preference optimization methods such as DPO \cite{rafailov2023direct} and its variants. As users provide feedback on agent behaviors, $\epsilon_t \to 0$, ensuring convergence to optimal matching by Theorem \ref{thm:convergence_full}.
\end{corollary}

\section{MDP Formulation Background}
\label{app:mdp}

We view each interaction as a \emph{multi-agent MDP} $\mathcal{M}=(\mathcal{S}, \mathcal{A}_i, \mathcal{A}_j, T, R, \gamma)$ in which two policies $\pi_i,\pi_j$ act over states $s\in\mathcal{S}$ with transition dynamics $T$ and per-step rewards summarized by a compatibility signal $R$. Compatibility emerges from the returns induced by the joint policies. Since true human policies $\pi_i^*,\pi_j^*$ are unobserved, we approximate them with LLM-driven policies $\hat{\pi}_i,\hat{\pi}_j$ and estimate $R$ from simulated trajectories, following inverse reinforcement learning principles \cite{ng2000algorithms}.

\section{Experimental Details}
\label{app:experimental_details}

\subsection{Datasets}

\textbf{Speed Dating.} 8,378 four-minute dates, 552 participants (2002-2004). Features: demographics, interest ratings (17 activities), self-ratings (attractiveness, sincerity, intelligence, fun, ambition), partner preferences. Outcome: mutual match decision.

\textbf{Divorce Prediction.} 170 couples, 54 Likert-scale questions (1-5) covering conflict resolution, shared values, mutual understanding, and communication patterns \cite{yontem2019divorce}. Outcome: married (1) vs divorced (0).

\subsection{Implementation}

\textbf{Speed Dating:} Gemini 2.5 Flash Lite generates personas (300-500 words) from structured profiles. Mistral-Nemo simulates conversations (temperature=0.6). We evaluate three LLM methods:
\begin{itemize}
    \item \textbf{LLM Participant:} Each agent rates the date from their own perspective after the simulated conversation.
    \item \textbf{LLM Observer:} An external LLM analyzes the conversation transcript and provides a compatibility assessment based on interaction quality, mutual engagement, and conversational flow.
    \item \textbf{LLM Mixed:} Combines participant ratings ($r_1, r_2$) and observer rating ($r_3$) using learned weights $\mathbf{w} = [w_1, w_2, w_3]$ optimized on training data to predict match outcomes.
\end{itemize}
The Observer method can be enhanced with in-context learning (10 examples) for calibration. Baselines: logistic regression and cosine similarity.

\textbf{Divorce Prediction:} Gemini 2.5 Flash Lite generates couple personas emphasizing conflict styles and values. We simulate reactions to personalized critical events (career conflicts, trust breaches, caregiver burdens). The Observer method uses 10 ICL examples to mitigate self-assessment bias—without calibration, LLMs tend to be overly negative about relationship outcomes (especially the participants methods). Baseline: logistic regression on survey features.

\subsection{System Workflow and Prompts}
\label{app:workflow_prompts}

We provide concrete examples of our system's workflow and prompt templates to aid reproducibility and understanding of the LLM text world engine architecture.

\textbf{Workflow Overview:} The system operates in three phases: (1) \textit{Persona Generation} converts structured profile data into natural language narratives, (2) \textit{Interaction Simulation} uses dual-mode LLMs (as agents + environment) to generate conversation trajectories $\tau_{ij}$, and (3) \textit{Compatibility Assessment} extracts ratings from participant and observer perspectives.

\vspace{0.3cm}
\noindent\textbf{Phase 1 - Persona Generation (Speed Dating):}

\begin{tcolorbox}[colback=gray!5, colframe=gray!60, title=Persona Generation Pipeline, fonttitle=\bfseries\small]
\small
\textbf{Input:} Structured profile (age, gender, interests, preferences)

\textbf{Prompt:} ``Generate a 300-500 word persona narrative for a speed dating participant based on: [profile data]. Include: personality traits, life goals, dating preferences, and conversational style.''

\textbf{Output:} Natural language persona ($P_i$)
\end{tcolorbox}

\vspace{0.3cm}
\noindent\textbf{Phase 2 - Conversation Simulation (Speed Dating):}

\begin{tcolorbox}[colback=blue!5, colframe=blue!60, title=Agent Prompt Template, fonttitle=\bfseries\small]
\small
\textbf{System:} ``You are Person A in a speed dating session.\\
\hspace*{2em} Your persona: [P\_i]\\
\hspace*{2em} Reply format:\\
\hspace*{2em} <INNER\_THOUGHT>[private feelings]</INNER\_THOUGHT>\\
\hspace*{2em} <RESPONSE>[what you say]</RESPONSE>''

\textbf{User:} [Partner's previous message from $P_j$'s RESPONSE]

\textbf{Simulation Process:}
\begin{enumerate}
\item Person 1 opens (system prompt + opening instruction)
\item For rounds 1..N:
\begin{itemize}
\item Person 2 responds to Person 1's public RESPONSE
\item Person 1 responds to Person 2's public RESPONSE
\end{itemize}
\item INNER\_THOUGHTs remain private (not shared)
\end{enumerate}
\end{tcolorbox}

\vspace{0.3cm}
\noindent\textbf{Phase 3 - Observer Evaluation (Speed Dating):}

\begin{tcolorbox}[colback=green!5, colframe=green!60, title=Observer Evaluation Prompt, fonttitle=\bfseries\small]
\small
\textbf{Prompt:} ``You are a relationship psychologist. Evaluate this couple based on established theories (Similarity-Attraction, Social Exchange, Attachment Theory).

Person 1: [P\_1 narrative]\\
Person 2: [P\_2 narrative]\\
Conversation: [full transcript of RESPONSEs only]

Rate compatibility (0-10) considering:
\begin{itemize}
\item Shared interests/values
\item Communication quality
\item Mutual attraction signals
\item Long-term potential''
\end{itemize}

\textbf{In-Context Learning (ICL):} 10 calibration examples shown:\\
\hspace*{1em}``Example 1: [P\_1 bg], [P\_2 bg] $\rightarrow$ Match/No Match''\\
\hspace*{1em}... (5 matches, 5 non-matches)
\end{tcolorbox}

\vspace{0.5cm}
\noindent\textbf{Critical Events Workflow (Divorce Prediction):}

The divorce task extends the architecture with a \textit{world engine} that maintains environmental state while agents respond from their personas. This separation ensures agents never "speak for each other"—the engine only describes external circumstances.

\vspace{0.3cm}
\noindent\textbf{Phase 1 - Persona Generation (Divorce):}

\begin{tcolorbox}[colback=gray!5, colframe=gray!60, title=Divorce Persona Generation, fonttitle=\bfseries\small]
\small
\textbf{Input:} 54 Gottman DPS survey responses (0-4 Likert scale: trust, communication, conflict resolution, shared values)

\textbf{Prompt:} ``Generate a 300-500 word persona for a married individual based on survey data. Emphasize: (1) conflict resolution style, (2) trust patterns, (3) core values, (4) communication approach, (5) boundaries and deal-breakers.''

\textbf{Output:} Husband persona ($P_h$) + Wife persona ($P_w$)
\end{tcolorbox}

\vspace{0.3cm}
\noindent\textbf{Phase 2 - World Engine Initialization:}

\begin{tcolorbox}[colback=orange!5, colframe=orange!60, title=World Engine System Prompt, fonttitle=\bfseries\small]
\small
\textbf{Role:} LIFE CIRCUMSTANCES narrator for a married couple

\textbf{Constraints:}
\begin{itemize}
\item Describe realistic scenarios and environmental changes ONLY
\item NEVER speak for husband or wife (no ``he says'' or ``she thinks'')
\item After agents respond, describe how environment/situation evolves
\item Add realistic stakes: time pressure, resource constraints, irreversible consequences
\end{itemize}

\textbf{Initial Scenario Prompt:}\\
``Critical Event: [Generated scenario from ICL, e.g., career conflict]\\
Husband Persona: [$P_h$]\\
Wife Persona: [$P_w$]\\
Set the scene with concrete sensory details (where they are, atmosphere, initial moment).''

\textbf{Example Output:}\\
``It's 9 PM on a Thursday. The dinner table is still messy. Your spouse just got off a two-hour call with their boss, eyes bright with excitement. They turn to you and say they need to talk about 'our future.' The job offer letter sits between you—2000 miles away, starting in 6 weeks.''
\end{tcolorbox}

\vspace{0.3cm}
\noindent\textbf{Phase 3 - Agent Interaction (with Environment Feedback):}

\begin{tcolorbox}[colback=blue!5, colframe=blue!60, title=Agent Response Template (Divorce), fonttitle=\bfseries\small]
\small
\textbf{Agent Prompt (Husband/Wife):}\\
``You are the [husband/wife] in this marriage.\\
Your persona: [$P_h$ or $P_w$]

\textbf{ICL Examples:} {[}5 divorced + 5 married couples' survey responses + outcomes{]}

\textbf{Current Situation:}\\
{[}Latest environment description from world engine{]}\\
{[}Spouse's last public RESPONSE{]}

\textbf{Format:}\\
<INNER\_THOUGHT>\\
{[}Self-check: How does your reaction align with your persona?{]}\\
{[}Your private true feelings—not what you ``should'' think{]}\\
</INNER\_THOUGHT>

<RESPONSE>\\
{[}What you actually say or do{]}\\
</RESPONSE>''

\textbf{Example Agent Output:}\\
<INNER\_THOUGHT>\\
My persona emphasizes independence (Atr12=1). This feels like my life being decided without me. I'm angry but don't want to seem unsupportive. Conflict avoidance (Atr3=2) makes me want to say yes, but my boundary is being crossed.\\
</INNER\_THOUGHT>

<RESPONSE>\\
``I... I'm happy for you, but this is a lot to process. Can we talk about what this means for my career? I've worked five years to get where I am.''\\
</RESPONSE>
\end{tcolorbox}

\vspace{0.3cm}
\noindent\textbf{Phase 4 - Environment Evolution (World Engine Reacts):}

\begin{tcolorbox}[colback=orange!5, colframe=orange!60, title=World Engine Reaction Prompt, fonttitle=\bfseries\small]
\small
\textbf{Input:}\\
- Husband's last RESPONSE: [public action]\\
- Wife's last RESPONSE: [public action]\\
- Couple personas: [$P_h$, $P_w$]\\
- ICL references: [Same 10 example couples to maintain consistency]

\textbf{Prompt:}\\
``Never speak for the husband or wife. Evolve the environment and situational pressures only. How does the atmosphere shift? What external changes occur? (2-3 sentences)''

\textbf{Example Output:}\\
``The silence stretches. Your spouse's excitement dims—they look away, jaw tight. The job offer deadline looms: 48 hours to decide. Outside, rain starts tapping the window. The unspoken question hangs: whose dream matters more?''
\end{tcolorbox}

\vspace{0.3cm}
\noindent\textbf{Phase 5 - Observer Evaluation (with ICL):}

\begin{tcolorbox}[colback=green!5, colframe=green!60, title=Observer Evaluation with Survey Calibration, fonttitle=\bfseries\small]
\small
\textbf{Prompt:} ``You are a relationship psychologist. You have 10 reference couples' survey responses (27 questions, 0-4 scale) with known outcomes.

\textbf{Reference Couples:}\\
{[}5 divorced (score 6.0-9.0): Low trust/communication/shared values{]}\\
{[}5 married (score 1.0-4.5): High trust/communication/shared values{]}

\textbf{Target Couple:}\\
Survey: [54 Gottman DPS responses in Q\&A format]\\
Interactions: [Public RESPONSEs from 6 rounds across 3 critical events]

\textbf{Task:} Compare target's survey + behavior with references. Assign divorce likelihood (0.0-10.0).

\textbf{Output:}\\
<ANALYSIS>\\
{[}Which references match? Key risk/stability signals?{]}\\
</ANALYSIS>

<SCORE>x.x</SCORE>''
\end{tcolorbox}

\vspace{0.3cm}
\noindent\textbf{Key Architectural Principles:}

\begin{itemize}
\item \textbf{Separation of Concerns:} World engine manages environment state; agents manage internal policies
\item \textbf{Private vs. Public:} INNER\_THOUGHT (persona-driven, private) vs. RESPONSE (observable by spouse)
\item \textbf{ICL Consistency:} Same 10 reference couples used across persona generation, agent prompts, environment evolution, and observer evaluation to maintain coherent simulation
\item \textbf{No Mind Reading:} Engine never assumes agent thoughts; agents never see spouse's inner thoughts
\end{itemize}

This architecture enables stress-testing relationships through pivotal moments while preserving persona authenticity and interpretability.

\end{document}